\documentclass[a4paper,11pt]{article}

\usepackage{amsfonts,amssymb,amsmath,amsthm,dsfont,xfrac,xspace}
\usepackage{fullpage}
\usepackage{graphicx}
\usepackage{subfig}
\usepackage{float}
\usepackage{cite}
\usepackage{hyperref}
\usepackage{algorithmic}
\graphicspath{{./},{fig/}}
\makeatletter
\let\NAT@parse\undefined
\makeatother
\usepackage[sort&compress, numbers]{natbib}
\usepackage{hyperref}

\newsavebox{\ieeealgbox}

\newtheorem{amsthm}{Theorem}

\newtheorem{amslem}{Lemma}

\title{MAC address anonymization for crowd counting} 

\author{Jean-Fran\c cois~Determe$^*$,
	Sophia~Azzagnuni$^*$,
	Fran\c cois~Horlin$^*$,\\
	and~Philippe~De~Doncker
	\thanks{All authors are with the OPERA Wireless Communications Group, Université libre de Bruxelles, 1050 Brussels, Belgium. Corresponding e-mail: Jean-Francois.Determe@ulb.be. Innoviris funded Jean-François Determe. }
}

\begin{document}
\maketitle

\begin{abstract}
Research has shown that counting WiFi packets called probe requests (PRs) implicitly provides a proxy for the number of people in an area. In this paper, we discuss a crowd counting system involving WiFi sensors detecting PRs over the air, then extracting and anonymizing their media access control (MAC) addresses using a hash-based approach. This paper discusses an anonymization procedure and shows time-synchronization inaccuracies among sensors and hashing collision rates to be low enough to prevent anonymization from interfering with counting algorithms. In particular, we derive an approximation of the collision rate of uniformly distributed identifiers, with analytical error bounds.
\end{abstract}

\section{Introduction} \label{sec:intro}
Many an event organizer deals with crowd monitoring and management \cite{martella2017current}. Recently, works from different teams proposed crowd counting systems using WiFi signals \cite{uras2020pma, determe2020forecasting, singh2020crowd}. These works describe counting systems detecting special control packets of the WiFi protocol: probe requests (PRs). Such packets are periodically transmitted by WiFi user terminals to detect nearby access points. Therefore, a PR-based counting system eludes the need for user cooperation and the need for an active WiFi connection from terminals to access points within range.

Typically, several WiFi sensors are deployed over the monitored area to detect PRs---and then extract and anonymize their media access control (MAC) addresses. Sensors finally timestamp anonymized PRs and transmit them to a central server processing them jointly. The number of distinct PRs acquired during a time frame of $T$ seconds (with $T = 60 s$ in this paper) implicitly provides a \textit{rate of PR transmission}, which is proportional (in average) to the number of attendees (as shown experimentally in \cite{determe2020forecasting} and theoretically in \cite{determe2022monitoring}). The proportionality between what is measured (the rate of PR transmission) and what is interesting to event organizers (the number of attendees) is referred to as the \textit{extrapolation factor} in our previous works and is determined experimentally.

Only in circumstances where occupation varies significantly in 60 seconds is our system less accurate (because it averages probe requests over a time frame of one minute, thereby smoothing any occupancy change occurring over a one minute time frame). However, both indoor and outdoor measurements in \cite{determe2022monitoring, determe2020forecasting} indicate it does not seem common for events or buildings hosting at least a few hundreds of individuals, probably because they enter and leave monitored areas at different times and also because entry and exit points have limited flow capacity.

\subsection{A short description of the monitoring system architecture}

\begin{figure}
	\centering
	\includegraphics[scale=0.70]{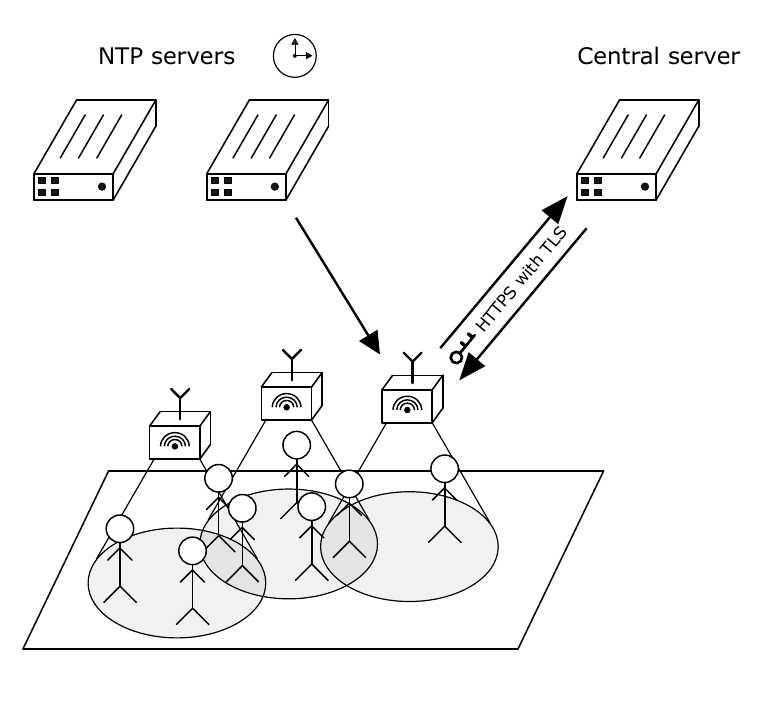}
	\caption{Scheme of the PR sensing procedure. Three WiFi sensors with overlapping ranges detect WiFi probe requests emitted by the smartphones of individuals. The shaded ellipses and the associated cones depict sensor detection ranges. Each sensor uses HTTPS links to periodically retrieve server peppers from the central server and uses another HTTPS link to upload anonymized PRs. Time synchronization is achieved by calibration with NTP servers. Communication links are depicted for only one sensor, to avoid clutter.}
	\label{fig:globalScheme}
\end{figure}

Figure~\ref{fig:globalScheme} depicts the experimentally validated counting scheme in \cite{determe2020forecasting, singh2020crowd}. Sensors (three in Figure~\ref{fig:globalScheme}) monitor an area and, because their effective detection range is not known precisely (it depends on the propagation environment and decreases as the density of people increases because of body-induced attenuation), they are usually installed densely enough to make detection ranges overlap. Data transfers between sensors and the central server are secured using hypertext transfer protocol secure (HTTPS) connections (with transport layer security (TLS)) so that the traffic is encrypted and the identity of the central server is verified---the latter preventing man-in-the-middle attacks. Sensors synchronize theirs clocks using network time protocol (NTP) servers.

\subsection{Collected data and the anonymization procedure} \label{subsec:dataAnonymization}

As depicted in Figure~\ref{fig:anonymProcedureSensor}, sensors extract three key data from each PR: i) a timestamp (whose precision is of one second), ii) a received signal strength indicator (RSSI) in dBm and iii) a source address (SA) (MAC address). Although some smartphones randomize the SAs embedded in PRs, it is not guaranteed and we want user tracking  to remain impossible, even without terminal-side SA randomization. Thus, we transform the original SA into an \textit{SA identifier}, which is its anonymous counterpart.

To generate an \textit{SA identifier} from an SA, we use a SHA-256 hash function in conjunction with a pepper and truncate its ouput to 64 bits. With $\lbrace 0, 1 \rbrace^\gamma$ denoting the set of all binary sequences of $\gamma$ bits, our anonymization function is $h: \mathcal{X} \rightarrow \lbrace 0, 1 \rbrace^{64}$, which is a truncated SHA-256 hash function whose inputs are 48-bit SAs ($\mathcal{X} = \lbrace 0, 1 \rbrace^{48}$). Note that generating SA identifiers of 64 bits is advantageous as such binary sequences can be easily stored as long integers in most databases (e.g., using the standard SQL \texttt{BIGINT} data type).

We prepend a time-varying pepper to every MAC address before hashing it. With $||$ denoting the concatenation operation, and \texttt{mac\_address} and \texttt{global\_pepper} representing respectively the MAC address (i.e., the SA) to be anonymized and the pepper prepended, $h(\texttt{global\_pepper} || \texttt{mac\_address})$ generates the SA identifier. 

\begin{figure}[h]
	\centering
	\includegraphics[scale=0.65]{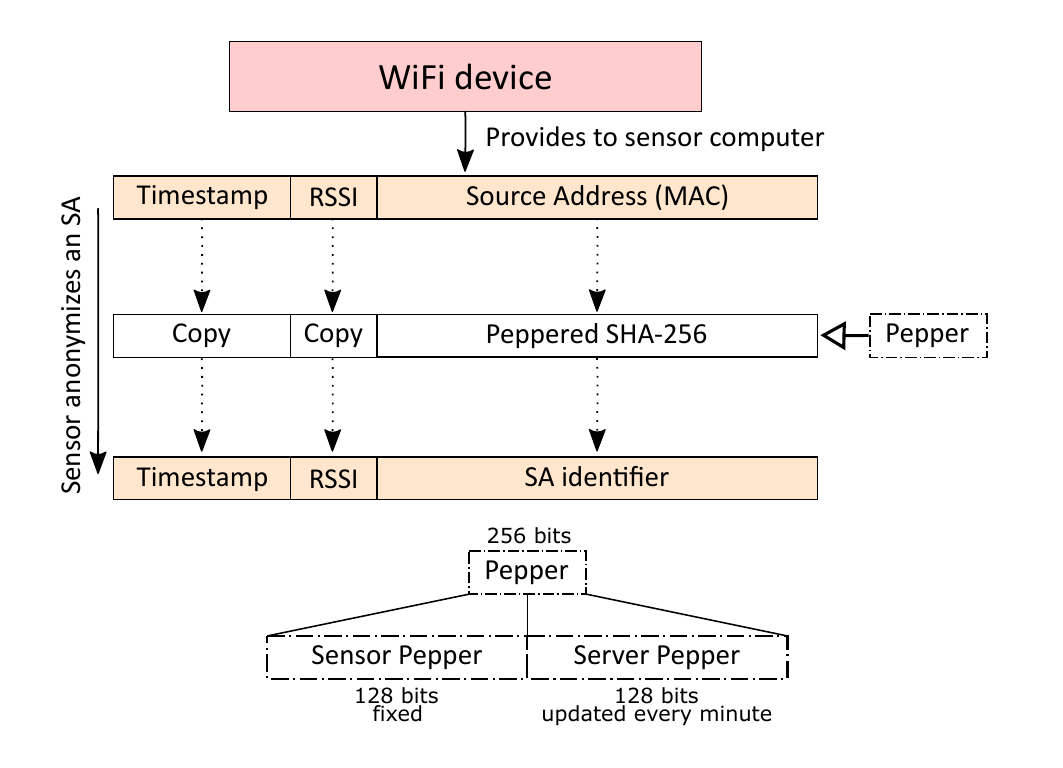}
	\caption{(From \cite{determe2022monitoring}) Scheme of the anonymization procedure executed by sensors.}
	\label{fig:anonymProcedureSensor}
\end{figure}

The pepper consists in a concatenation of a fixed 128-bit \textit{sensor pepper} and a time-varying 128-bit \textit{server pepper}. The central server maintains an up-to-date array of 20 server peppers for a duration of 20 minutes that sensors periodically fetch using an HTTPS link with transport security layer (TLS). Sensors use each server pepper for a specific \textit{one-minute time frame}. Server peppers are generated using a pseudo random number generator (PRNG) (e.g., \texttt{/dev/urandom} or \texttt{/dev/random} on Linux). If this PRNG is deemed not secure (see \cite{dodis2013security}), hardware PRNG generators are alternatives \cite{stipvcevic2007quantum, zheng20196}.

The server and the sensors delete server peppers once they become outdated---in particular, the sensors erase the volatile memory chunk storing server peppers before updating it with new peppers periodically retrieved from the server.

The fixed sensor pepper forms a last line of defense in case the server peppers get compromised. It is written in a file or in the codebase of the sniffer, and it is never stored on the server. We proposed a fixed sensor pepper but storing pregenerated sensor peppers for time frames of one minute is possible too; it would represent 42 MB of data for five years.

Loosely speaking, the time-varying and eventually forgotten pepper has a high entropy and breaking the anonymization scheme is about finding out its value for all one-minute time frames of interest. As we show thereafter, this procedure is not computationally tractable. We also explain why SA identifiers generated using different peppers cannot be compared against one another, thereby precluding user tracking. Moreover, despite the data distortion that our anonymization procedure entails, we also demonstrate that it does not affect in any significant way the output of our counting method (a procedure that we explained in more details in \cite{determe2020forecasting, determe2022monitoring}). Intuitively, anonymization cannot affect counting if the SA identifier of any SA is identical across all sensors at (almost) all time instants.

It is also possible for $h: \mathcal{X} \rightarrow \lbrace 0, 1 \rbrace^{64}$ to output (random) tokens, instead of being a truncated cryptographic hash function. In this case, the outputted tokens are truly uniformly distributed in the space $\lbrace 0, 1 \rbrace^{64}$. The associated tokens should be kept in volatile memory (as well as the corresponding inputs) for a given anonymization window but can be wiped out once a new anonymization window begins. This random-token approach is typically well suited to a central and final anonymization round. It would not be practical to carry it out on a distributed network because all nodes should then agree on a mapping from input SAs to tokens in real time.

\subsection{Contributions}

Previous sections review the crowd monitoring system used in \cite{determe2020forecasting, singh2020crowd} for forecasting purposes and presented  in \cite{determe2022monitoring} (with fewer details about anonymization than in this manuscript). This paper discusses the strength of our anonymization procedure and the effect of time synchronization inaccuracies on it. Besides the proposal of the anonymization process, our contributions also include the demonstration that our system satisfies the following four requirements:
\begin{enumerate}
	\item It is computationally intractable to recover the original MAC addresses from the anonymous identifiers our system generates.
	\item Anonymous identifiers from two distinct one-minute time frames cannot be compared against one another, which entails the impossibility to track individuals over time.
	\item The proportion of time instants during which two sensors of our system could generate distinct anonymous identifiers for the same MAC address is negligible.
	\item Assuming WiFi devices in an area generate $10^7$ distinct MAC addresses within one minute in a monitored area, the collision rate of our anonymization procedure is lower than $10^{-9}$. The value of $10^7$ distinct MAC addresses corresponds roughly to an event of a few million people, which is comparable to or higher than the number of attendees of the vast majority of public events in the world.
\end{enumerate}

Requirements 1) and 2) guarantee privacy, in that the original MAC addresses of devices cannot be recovered and also because tracking individuals is impossible. Requirements 3) and 4) enable the central server to compute accurate attendee counts. Should Requirement 3) not be met, sensors would too often return different SA identifiers for identical devices simultaneously detected (because of overlapping detection ranges), thereby inducing a positive counting bias. Requirement 4) ensures a negligible probability of two devices being identified as a single one (which would imply a negative counting bias).

Proving our system meets Requirement 4) is overwhelmingly a mathematical effort that is based on mathematical approximations of the collision rate of hash functions. This is the most complex result to derive in this paper and, due to its general nature, the theorem approximating the collision rate could be of interest to researchers pursuing other endeavors than the design of a crowd counting system.

\subsection{Comparison with the state of the art} \label{sec:soa}

The authors of \cite[Sec.~5]{demir2014analysing} succinctly mentioned using random binary sequences appended to the MAC addresses prior to hashing (or to replace MAC addresses with tokens, more specifically, universally unique Identifiers (UUIDs) \cite{leach2005universally}). Our anonymization scheme uses a similar idea, except that we prepend random sequences a central server partially generates and then shares with time-synchronized sensors. Each sequence is used simultaneously by all our sensors during one minute, a time after which the server and the sensors erase it. Thus, brute force attacks consist in recovering a pepper of high entropy instead of hashed MAC addresses, whose entropy is too low to withstand such attacks \cite{demir2014analysing, demir2017pitfalls, marx2018hashing}. We also split peppers into two parts (which \cite{demir2014analysing} does not propose), with one unknown to the server.

In \cite{fuxjaeger2016towards}, the authors develop a system similar to ours but for road traffic monitoring. Their anonymization scheme \cite[Sec.~VI]{fuxjaeger2016towards} relies on a truncation of the MAC address prior to hashing, whereas we rely on time-varying peppers of sufficiently high entropy to ensure anonymity and prevent brute-force attacks. Based on their experiments, it is unclear whether their anonymity scheme based on MAC address truncation would yield unacceptably high collision rates for large-scale crowds. 

The very recent work \cite{ali2020practical} presents research similar to ours. \cite{ali2020practical} derives the collision rate we present in Theorem~\ref{thm:exactCollProp} \cite[Sec.~4.2]{ali2020practical} and also justifies the interest of such a derivation within the framework of WiFi and Bluetooth signal detection for crowd counting. They also validate Theorem~\ref{thm:exactCollProp} numerically for a number of MAC addresses lower than $2 \; 10^5$ and for a number of output bits after hashing of up to 24 bits \cite[Sec.~1]{ali2020practical}. Thanks to our precise approximation of the collision rate (see Theorem~\ref{thm:approxThm}), we can handle vastly higher values (e.g., a number of output bits of 64 bits and $10^7$ MAC addresses). Moreover, our method is based on secret peppers that are forgotten and that are split in two parts: one stored on sensors and the other stored on a central server, so that if either the sensors or central server are comprised, anonymity still holds (see Section~\ref{subsec:req1}). The time-varying nature of our peppers also makes it impossible to track individuals (see Section~\ref{subsec:req2}). We also discuss the impact of typical time synchronization errors on modern networks and find them to have no significant impact on the counting process we used in \cite{determe2020forecasting, singh2020crowd} (see Section~\ref{subsec:req3}). Finally, we point out that our (novel) approximation of the collision rate (and its analytical error bounds) are non-trivial mathematical results to derive (see Section~\ref{subsec:req3} and the Appendix).

Other related works on crowd counting using WiFi probe requests are \cite{hong2018crowdprobe} and \cite{potorti2018localising}. In particular, \cite{hong2018crowdprobe} discusses smartphone-executed MAC address randomization and its impact on crowd counting algorithms. The authors also propose a 	method for generating fingerprints that allow them to track individuals whose smartphones emit PRs (a possibility that our system precludes on purpose for privacy reasons). The work \cite{potorti2018localising} deals with user positioning, especially in indoor environments and for non-dense crowds. They notably improve positioning accuracy by leveraging signal strength indicators.
\subsection{Outline}

Section~\ref{sec:intro} has detailed the way our system works, with Section~\ref{sec:soa} comparing our results against the state of the art. Then, Section~\ref{sec:fourRequirements} shows that our four requirements are met. Finally, Section~\ref{sec:conclusion} is the conclusion. The Appendix contains mathematical proofs.

\section{Results} \label{sec:fourRequirements}
We now turn to our contribution: proving our four requirements are met by the already existing crowd counting system presented in \cite{determe2020forecasting, singh2020crowd, determe2022monitoring}. We insist again that these results are new and not detailed in \cite{determe2020forecasting, singh2020crowd, determe2022monitoring}.

The data collection process is a means to an end: make it possible to count the number of people visiting an area while ensuring their privacy. In other words, the need for satisfying the four requirements is about ensuring two properties: privacy and accurate counting. The first two requirements address the former: how to ensure the privacy of users is preserved and tracking them (even anonymously) is impossible? The penultimate and last requirements deal with the second property: how to ensure that our privacy-enhancing data distortion does not affect counting accuracy? The next subsections detail our four requirements and show how our system satisfies them.

\subsection{Requirement 1:  impossibility to recover the original SA from SA identifiers} \label{subsec:req1}

Cryptographic hash functions like SHA-256 cannot be directly reversed---in practice, reversing consists in trying inputs until finding one whose hash is the output to be reversed. It is possible for an attacker to know the input MAC address of a particular entry in the list of anonymized PRs; for example, an attacker may go near sensors and send fake PRs with precise timing patterns that make it easy to identify them. In this case, brute forcing the pepper entails testing many of the 256-bit sequences that exist (on average, half of them should be tested). Attackers usually perform this operation using graphical processing units (GPUs), field-programmable gate arrays (FPGAs), or, if they have large resources, application-specific integrated circuits (ASICs). Let us examine if this attack is feasible with GPUs.

For example, 1 million Nvidia RTX 2080 SUPER Founders Edition graphics cards can compute roughly 5700 SHA-256 TeraHashes per second \cite{rtx2080fesha256benchmark}---this implies that testing all 256-bit peppers (approximately $1.16\; 10^{65}$ TeraHashes) takes $2.04 \; 10^{61}$ seconds, i.e., $6.47 \; 10^{53}$ years. Should one of the two 128-bit peppers be known to an attacker, testing all 128-bit sequences still takes roughly $1.90 \; 10^{15}$ years. We point out that relying on a regular SHA-256 hash function without peppers is not safe (see \cite{demir2014analysing, marx2018hashing} and \cite[Sec.~VI]{demir2017pitfalls}) as the entropy of MAC addresses is too low to resist brute force attacks. We also highlight that using computationally intensive hashes like \textit{bcrypt} \cite{provos1999future} and \textit{Argon2} \cite{biryukov2016argon2} would imply unreasonable computational requirements for sensors (see also \cite[Sec.~5]{demir2014analysing}).

\subsection{Requirement 2: preventing tracking for more than one minute} \label{subsec:req2}

This requirement is linked to server peppers being updated between consecutive time frames of one minute. In particular, the avalanche effect of SHA-256 hash functions makes hashing with different peppers return incomparable SA identifiers for any fixed MAC address. (The avalanche effect of cryptographic hash functions is the fact that minor changes in the input significantly change the hash.)

\subsection{Requirement 3: peppers are identical across all sensors at a given time instant} \label{subsec:req3}

This requirement depends on the accuracy of time synchronization. We propose to use network time protocol (NTP), which implies accurate time synchronization on low-latency networks (e.g., 4G networks, with timing errors lower than 10 ms \cite{mivskinis2014timing}). There could be synchronization-related mismatches at the frontiers of consecutive one-minute time frames but only for 20 ms/60000 ms = 0.033 \% of their duration. Assuming probe request transmission times are uniformly distributed in time, this figure translates into having on average 0.033 \% of all PRs being anonymized by different peppers on the sensors.

\subsection{Requirement 4:  a collision rate of less than $10^{-9}$ for $10^7$ MAC addresses} \label{subsec:req4}

We now derive estimates of the collision rate of truncated hash functions. The first part of this section is mathematical while the second part leverages the results of the first one to show the collision rate achieved by our system to be negligible for up to 10 million SAs.

\subsubsection{Mathematical foundations}

Variable $m$ denotes a number of possible outputs, such that $\log_2(m) \in \mathbb{N}$, and  $\lbrace 0, 1 \rbrace^\gamma$ denotes the set of all binary sequences of $\gamma$ bits. We consider a function $h: \mathcal{X} \rightarrow \lbrace 0, 1 \rbrace^{\log_2(m)}$ (with $n := \mathrm{card}(\mathcal{X})$). Hereafter, $h$ is a hash function, whose output is approximately uniformly distributed in $\lbrace 0, 1 \rbrace^{\log_2(m)}$~\cite[Sec.~9.7.1]{menezes1996handbook}. It could also be a token generator, in which case the uniform distribution assumption is exactly satisfied.

We follow the standard terminology in the study of hash tables and refer to $m$ and $n$ as the \textit{number of buckets} and the \textit{number of inserts}, respectively. Similarly, $\alpha := n/m$ is called the \textit{load factor}. Finally, $Y^{(n,m)}$ denotes the (random) number of collisions when inserting $n$ values into $m$ buckets (with the uniform distribution assumption). Theorem~\ref{thm:exactCollProp} provides an exact---yet numerically unstable---formula of $\mathbb{E} \left\lbrack Y^{(n,m)} \right\rbrack$.

\begin{amsthm} \label{thm:exactCollProp}
	For $n$ inserts into $m$ buckets, the collision rate, $\mathbb{E} \lbrack Y^{(n,m)}\rbrack/n$, is
	\begin{equation} \label{eq:exactCollRate}
		\dfrac{\mathbb{E} \left\lbrack Y^{(n,m)} \right\rbrack}{n} = 1 - \dfrac{m}{n} \left( 1 - \left( \dfrac{m-1}{m} \right)^n \right),
	\end{equation}
	where the uniform distribution assumption has been used.
\end{amsthm}
\begin{proof}
	See the Appendix.
\end{proof}

As shown in Figure~\ref{fig:HashCollisionApproxResEq1}, (\ref{eq:exactCollRate}) suffers from numerical instabilities for sufficiently low values of the load factor. Therefore, for systems whose load factors are too low for (\ref{eq:exactCollRate}) to provide accurate estimates, approximations are needed. In particular, to ensure such approximations are accurate enough, they should have proven analytical error bounds. Theorem~\ref{thm:approxThm} proposes three approximations of $\mathbb{E} \lbrack Y^{(n,m)}\rbrack/n$, with proven error bounds. Only the penultimate and last inequalities of Theorem~\ref{thm:approxThm} are numerically stable.

\begin{figure}[h]
	\centering
	\includegraphics[scale=0.75]{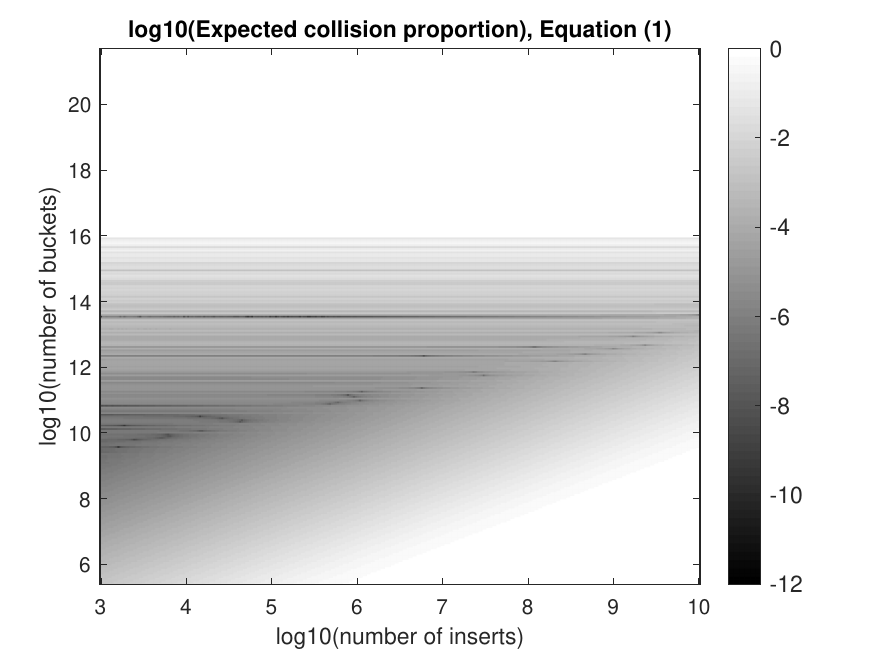}
	\caption{Numerically computed value of $\log_{10} (\mathbb{E} \left\lbrack Y^{(n,m)} \right\rbrack / n)$ (using (\ref{eq:exactCollRate})) in Matlab R2019a as a function of the number of inserts $n$ and the number of buckets $m$. With $\log_{10}(n) \geq 3$, numerical instabilities appear for values of $\log_{10}(m)$ as low as $9$.}
	\label{fig:HashCollisionApproxResEq1}
\end{figure}

\begin{amsthm} \label{thm:approxThm}
	For a degree of approximation $K \geq 2$, a number of inserts $n \geq 2$, and a load factor $\alpha \leq 1$, there exist error terms $\delta(\alpha, n)$ and $R_{K-1}(\alpha)$ such that
	\begin{align}
		\dfrac{\mathbb{E} \left\lbrack Y^{(n,m)} \right\rbrack}{n} & = 1 - \alpha^{-1} \left( 1 - \exp(-\alpha) \right) + \delta(\alpha, n) \label{eq:1stapprox} \\
		& = \sum_{k = 1}^{K-1} \dfrac{\alpha^k (-1)^{k+1}}{(k+1)!} + \delta(\alpha, n) + R_{K-1}(\alpha) \label{eq:2ndapprox} \\
		& = \dfrac{\alpha}{2} + \delta(\alpha, n) + R_1(\alpha), \label{eq:3rdapprox}
	\end{align}
	where
	\begin{equation} \label{eq:deltaIneq}
		- \sqrt{\dfrac{\alpha^2}{n^2 - \alpha^2}  \left(\dfrac{\pi^2}{6} - 1\right) } \leq \delta(\alpha, n) \leq 0,
	\end{equation}
	\begin{equation} \label{eq:RIneq}
		|R_{K-1}(\alpha)| \leq \dfrac{\alpha^{K}}{(K+1)!},
	\end{equation}
	and, in particular,
	\begin{equation} \label{eq:ROneTermRelativeIneq}
		\dfrac{|R_1(\alpha)|}{\alpha/2} \leq \dfrac{\alpha}{3}.
	\end{equation}
\end{amsthm}
\begin{proof}
	See the Appendix.
\end{proof}

\subsubsection{The interpretation of Theorem~\ref{thm:approxThm}}

Theorem~\ref{thm:approxThm} approximates the exact value of the collision rate that Theorem~\ref{thm:exactCollProp} provides. Equation~(\ref{eq:1stapprox}) yields a first approximation that is not numerically stable for sufficiently low values of $\alpha$ (a figure similar to Figure~\ref{fig:HashCollisionApproxResEq1} can be easily generated for (\ref{eq:1stapprox}) but has been omitted for the sake of brevity). Equation~(\ref{eq:2ndapprox}) provides a numerically stable approximation whose precision is controlled through $K$, hence the name ``\textit{degree of approximation}''.

The error term $\delta(\alpha,n)$ quantifies to what extent $\left( 1 - \alpha / n \right)^n$ accurately approximates $\exp(- \alpha)$. The term $R_{K-1}(\alpha)$ bounds the error tied to approximating $\exp(- \alpha)$ using its $K$th-order Taylor polynomial, an approach used to derive~(\ref{eq:2ndapprox}) from~(\ref{eq:1stapprox}).

For low values of $\alpha$ (e.g., $\alpha \leq 10^{-3}$),~(\ref{eq:3rdapprox}) is an accurate approximation because $|R_1(10^{-3})|/(10^{-3}/2) \leq 10^{-3}/3$ (see (\ref{eq:ROneTermRelativeIneq})), i.e., the error $|R_1(\alpha)|$ is less than 0.1 \% of the approximated value $\alpha/2$. For $\alpha \leq 1$ and for $n$ high enough (say, $n \geq 100$), $\alpha^2/(n^2 - \alpha^2) \simeq \alpha^2/n^2 = 1/m^2$. Thus, with $m \geq 2^{64}$, $|\delta(\alpha, n)| \leq m^{-1} 0.8031 \leq  5\; 10^{-20}$.

\subsubsection{Proving requirement 4 is satisfied} \label{subsubsec:Req4IsMet}

We have $m = 2^{64} \simeq 1.84 \; 10^{19}$, which means that we truncate SHA-256 hashes to 64 bits. This corresponds to a load factor $\alpha = 10^7 (1.84)^{-1} 10^{-19} \simeq 10^{-12}$ for $n = 10^7$ MAC addresses. Figure~\ref{fig:HashCollRes} then shows that the collision rate expectation is approximately equal to $10^{-12.5}$. Note that, for $\alpha$ sufficiently low, (e.g., $\alpha \leq 10^{-3}$), the approximation becomes~(\ref{eq:3rdapprox}), which explains why the level sets in Figure~\ref{fig:HashCollRes} appear to be linear slopes.

We point out that approximation errors are negligible for our choice of parameters. Our load factor $\alpha \simeq 10^{-12}$ implies (for any $K \geq 2$) $|R_{K-1}(\alpha)| \leq 10^{-24}$. Moreover, as already pointed out, $m \geq 2^{64} \implies |\delta(\alpha, n)| \leq 5 \; 10^{-20}$. 

The conclusion is that our estimate of the collision rate expectation is approximately equal to $10^{-12.5}$, with an error upper bounded by $5 \; 10^{-20} + 10^{-24} \simeq 5 \; 10^{-20}$, so that Requirement~4 is met.
\begin{figure}
	\centering
	\includegraphics[scale=0.75]{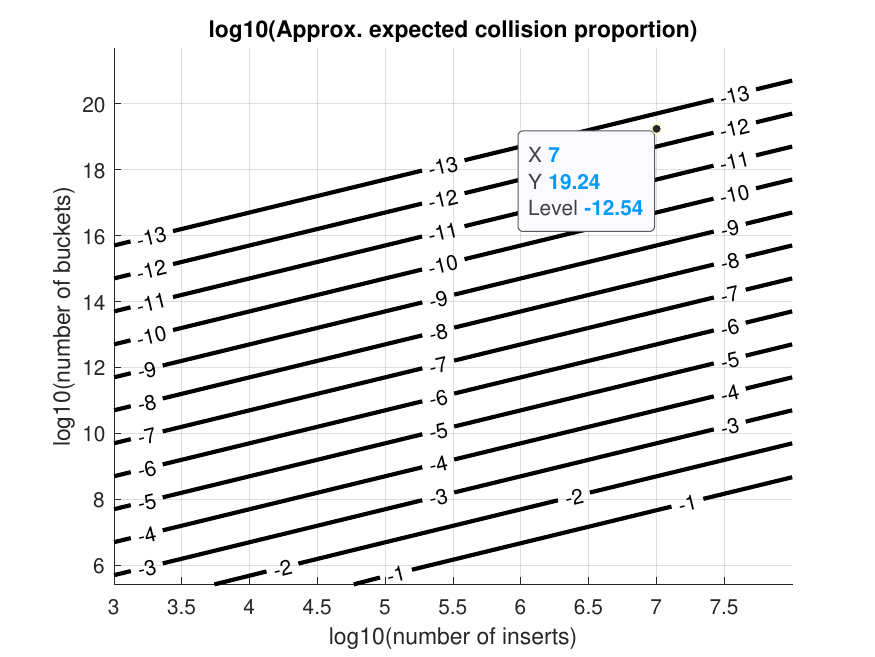}
	\caption{Levels sets of the approximation~(\ref{eq:2ndapprox}) of the collision rate as a function of the number of inserts $n$ and the number of buckets $m$.}
	\label{fig:HashCollRes}
\end{figure}

\subsubsection{Concentration inequality for the collision rate}

While it is interesting to upper bound the expectation of the collision rate, $Y^{(n,m)}/n$, upper bounding the probability that it exceeds some threshold is also a worthy endeavor. We propose such a (coarse) inequality. Because $Y^{(n,m)}/n \geq 0$, we can apply Markov's inequality:
\begin{equation}
	\mathbb{P} \left\lbrack \dfrac{Y^{(n,m)}}{n} \geq a \right\rbrack \leq \dfrac{\mathbb{E} \left\lbrack Y^{(n,m)}/n \right\rbrack}{a}.
\end{equation}

Using Theorem~\ref{thm:approxThm} with $K = 2$, we only know that $\mathbb{E} \left\lbrack Y^{(n,m)}/n \right\rbrack = \alpha/2 + \delta(\alpha, n) + R_1(\alpha)$ where $\delta(\alpha,n) \leq 0$ and $R_1(\alpha) \leq \alpha^2/6$. Therefore, we can only use the slightly more pessimistic concentration inequality that is
\begin{equation}
	\mathbb{P} \left\lbrack \dfrac{Y^{(n,m)}}{n} \geq a \right\rbrack \leq \dfrac{\alpha/2 + \delta(\alpha, n) + R_1(\alpha)}{a} \leq \dfrac{\alpha/2 + \alpha^2/6}{a},
\end{equation}
where the term $\alpha^2/6$ is negligible in comparison to $\alpha/2$ for $\alpha$ sufficiently low (e.g., $\alpha \leq 10^{-3}$).

For example, let us consider again the previous calculation of Section~\ref{subsubsec:Req4IsMet} (with $n = 10^7$ MAC addresses, $m = 2^{64}$ and $\alpha = 10^{-12}$), which yielded $\mathbb{E} \lbrack Y^{(n,m)}/n \rbrack = \alpha/2 + \delta(\alpha, n) + R_1(\alpha) \simeq 10^{-12.5}$. Owing to $\alpha^2/6 \ll \alpha/2$ and with $a = 10^{-9}$,
\begin{equation}
	\mathbb{P} \left\lbrack \dfrac{Y^{(n,m)}}{n} \geq 10^{-9} \right\rbrack \leq  \dfrac{10^{-12.5}}{10^{-9}} = 10^{-3.5} \simeq 3.16\; 10^{-4},
\end{equation}
which shows that, with probability 99.968 \%, the collision rate of our counting system does not exceed $10^{-9}$.

Markov's inequality is coarse (and it may be possible to improve our result using a more sophisticated inequality) but, within the context of upper bounding the collision rate of our crowd counting system for large crowds, that inequality is sufficient to prove its collision rate does not exceed $10^{-9}$ with high probability for large crowds ($10^7$ MAC addresses per minute).

	\subsection{Validating requirement 4 experimentally}
	An interesting future work endeavor would be to validate requirement 4 experimentally and to evaluate how sharp the inequalities we obtained are. In particular, an interesting question is to determine to what the extent truncated SHA-256 hashes are close to being randomly distributed and how a discrepancy from uniformity translates into higher collision rates in our particular application. A conceptually simple analysis of this question could be carried out by generating a statistically significant number of random peppers and, for each pepper, to generate at least $10^{14}$ random SAs to evaluate the empirical collision rate (which we know should be around $10^{-12.54}$ according to Figure~\ref{fig:HashCollRes}, which explains why generating at least $10^{14}$ SAs is statistically sound). Recent simulation results related to this approach are available in \cite[Sec.~5]{ali2020practical}.
	
	Unfortunately, rigorously validating the collision rate experimentally using datasets of true SAs would require to monitor events gathering millions of individuals. Moreover, it would be impossible to know exactly how many people carry smartphones and when each smartphones send PRs. As a result, we propose a slightly weaker variant (that still requires significant efforts). First of all, one needs to identify randomization and PR emission patterns from modern smartphones in a controlled laboratory environment (or use existing results on typical PR generation processes in the literature, see \cite[Fig.~1]{hong2018crowdprobe}). This is equivalent to building a statistical distribution that accurately depicts the random process of modern smartphones generating PRs. Then, the methodology of the previous paragraph can be used with this distribution instead of a uniform one for SAs. The difficulties here mainly are about identifying PR transmission patterns for an extensive set of modern smartphones as well as evaluating what is the market share of each smartphone that is tested.

\section{Conclusion}\label{sec:conclusion}

Within the framework of WiFi-based crowd counting, this paper proposes an anonymization scheme for collected MAC addresses. This anonymization scheme is endowed with four desirable properties. First, it makes the recovery of original MAC addresses computationally intractable. Second, it precludes tracking capabilities. Third, it works properly as long as timing synchronization errors between nodes collecting MAC addresses is of the order of 10 ms, which is typically easy to attain on modern cellular networks. Fourth, it achieves a negligible collision rate between MAC addresses. This last point is supported by ample theoretical evidence. Although this paper is motivated by crowd counting applications, the methods and mathematical results could be of interest in other domains.

\appendix
\section[\appendixname~\thesection]{Proofs}

\renewcommand*{\theamslem}{\thesection\arabic{amslem}} 

In what follows, $\| \boldsymbol{x} \|_2$ denotes the $\ell_2$-norm of vector $\boldsymbol{x}$. The notation $(a_k)_{1 \leq k \leq K}$ is equivalent to the vector $(a_1, a_2, \dots, a_K)$ of size $K$.

\subsection[\appendixname~\thesubsection]{Proof of Theorem~\ref{thm:exactCollProp}}

Let $p_j$ denote the probability that the $j$th ($1 \leq j \leq m$) bucket be empty after $n$ inserts. All inserts have equal probabilities to fall within each bucket and whether an insert ends up in one bucket is independent of which buckets are already occupied. As a result, we have $p_j = ((m-1)/m)^n$. Indeed, for the $j$th bucket to be unoccupied, all $n$ inserts should end up in any of the other $m-1$ buckets and, for each insert, there is a probability $(m-1)/m$ that it ends up in any bucket except the $j$th one. The expectation of the number of empty buckets after $n$ inserts is equal to 
\begin{equation*}
	\sum_{j=1}^m \mathbb{E} \lbrack A_j \rbrack = \sum_{j=1}^m \left( \dfrac{m-1}{m} \right)^n  =  m \left( \dfrac{m-1}{m} \right)^n,
\end{equation*}
where $A_j = 1$ if the $j$th bucket is empty and equals 0 otherwise. Hence, the expectation of the number of occupied buckets is $m - m ((m-1)/m)^n$. Without any collision after $n$ inserts, there are exactly $n$ distinct occupied buckets. However, with $n_l < n$ distinct occupied buckets, there are $n - n_l$ collisions. As the number of collisions is equal to $n - \textrm{``number of occupied bucket''}$ the average number of collisions is $n - m(1 - ((m-1)/m)^n)$ and the proof is complete.
\subsection[\appendixname~\thesubsection]{Lemmas for Theorem~\ref{thm:approxThm}}

To prove Theorem~\ref{thm:approxThm}, we shall first derive two lemmas. Lemma~\ref{lem:approxExpLimit} quantifies to what extent $( 1 - \alpha/n )^n$ is a good approximation of $\exp(- \alpha)$.

\begin{amslem} \label{lem:approxExpLimit}
	For $n \geq 1$ and $\alpha < n$,
	\begin{equation} \label{eq:approxExpLimit}
		\left( 1 - \dfrac{\alpha}{n} \right)^n = \exp(- \alpha) F(\alpha, n),
	\end{equation}
	where
	\begin{equation} \label{eq:FanIneqApproxExpLimit}
		\exp \left( - \alpha^2 \sqrt{\dfrac{1}{n^2 - \alpha^2}  \left(\dfrac{\pi^2}{6} - 1\right) } \right) \leq F(\alpha, n) \leq 1.
	\end{equation}
\end{amslem}
\begin{proof}
	For $0 \leq \alpha/n < 1$, using the Maclaurin series of $\log ( 1 - x ) = - \sum_{k=1}^{\infty} x^k/k$ (valid for $|x| < 1$), we obtain
	\begin{align}
		\left( 1 - \dfrac{\alpha}{n} \right)^n & = \exp \left( n \log \left( 1 - \dfrac{\alpha}{n} \right) \right) \nonumber \\
		& = \exp \left(- n \sum_{k=1}^{\infty} \dfrac{(\alpha/n)^k}{k}  \right) \nonumber \\
		& = \exp \left(- \alpha \left( 1 + \sum_{k=1}^{\infty} \dfrac{(\alpha/n)^k}{k+1} \ \right)  \right) \label{eq:lemApproxExpBuf1},
	\end{align}
	where we have used
	\begin{align*}
		- n \sum_{k=1}^{\infty} \dfrac{(\alpha/n)^k}{k} & = - \sum_{k=1}^{\infty} \dfrac{\alpha^k}{n^{k-1} k} = - \alpha \left( 1 + \sum_{k=2}^{\infty}  \dfrac{\alpha^{k-1}}{n^{k-1} k} \right) =  - \alpha \left( 1 + \sum_{k=1}^{\infty}  \dfrac{\alpha^{k}}{n^{k} (k+1)} \right).
	\end{align*}
	Defining $f^{(K)} (\alpha, n) := \sum_{k=1}^{K} (\alpha/n)^k/(k+1)$, we have, $0 < f^{(1)} (\alpha, n) < f^{(2)} (\alpha, n) < \cdots$ so that if for all $K$, $f^{(K)} (\alpha, n) \leq \xi(\alpha, n)$, then $\sum_{k=1}^{\infty} (\alpha/n)^k/(k+1) \leq \xi(\alpha, n)$. The sum in $f^{(K)} (\alpha, n)$ is the inner product between vectors $((\alpha/n)^k)_{1 \leq k \leq K}$ and $(1/(k+1))_{1 \leq k \leq K}$. Cauchy-Schwarz inequality yields:
	\begin{equation*}
		f^{(K)} (\alpha, n) \leq \sqrt{\left\| \left(\dfrac{\alpha^k}{n^k}\right)_{1 \leq k \leq K} \right\|_2^2 \left\| \left(\dfrac{1}{k+1}\right)_{1 \leq k \leq K} \right\|_2^2}.
	\end{equation*}
	We have, using an asymptotic expression for geometric series,
	\begin{align*}
		\left\| \left(\dfrac{\alpha^k}{n^k}\right)_{1 \leq k \leq K} \right\|_2^2 & = \sum_{k=1}^K \left( \left( \dfrac{\alpha}{n} \right)^k \right)^2 \\
		& =  \sum_{k=0}^K \left( \left( \dfrac{\alpha}{n} \right)^2 \right)^k - 1 \\
		& \leq \sum_{k=0}^{\infty} \left( \left( \dfrac{\alpha}{n} \right)^2 \right)^k - 1 \\
		& = \dfrac{1}{1 - \alpha^2 / n^2} - 1 \\
		& = \dfrac{\alpha^2}{n^2 - \alpha^2}.
	\end{align*}
	Moreover,
	\begin{align*}
		\left\| \left(\dfrac{1}{k+1}\right)_{1 \leq k \leq K} \right\|_2^2 & = \sum_{k = 1}^{K+1} \dfrac{1}{k^2} - 1 \\
		& \leq \sum_{k = 1}^{\infty} \dfrac{1}{k^2} - 1 \\
		& = \zeta(2) - 1,
	\end{align*}
	where $\zeta(2)$ is Riemann zeta function evaluated at 2, which is equal to $\pi^2/6$. Therefore, we may use the upper bound
	\begin{equation}\label{eq:lemApproxExpXiDef}
		\xi(\alpha, n) := \alpha \sqrt{\dfrac{1}{n^2 - \alpha^2}} \sqrt{\dfrac{\pi^2}{6} - 1}.
	\end{equation}
	It is also easy to notice that $\sum_{k=1}^{\infty} (\alpha/n)^k/(k+1) \geq 0$ given that all the terms of the sum are positive.\\
	
	Injecting these results in~(\ref{eq:lemApproxExpBuf1}), we obtain
	\begin{align*}
		\left( 1 - \dfrac{\alpha}{n} \right)^n & = \exp \left( - \alpha \left( 1 + \lim_{K \rightarrow \infty} f^{(K)}(\alpha, n) \right) \right) \\
		& = \exp( - \alpha) \exp \left( - \alpha  \lim_{K \rightarrow \infty} f^{(K)}(\alpha, n) \right) \\
		& = \exp( - \alpha) F(\alpha, n)
	\end{align*}
	where
	\begin{align*}
		F(\alpha, n) \leq & \exp (0) = 1
	\end{align*}
	and
	\begin{align*}
		F(\alpha, n) \geq & \exp \left( - \dfrac{\alpha^2}{\sqrt{n^2 - \alpha^2}} \sqrt{\dfrac{\pi^2}{6} - 1} \right)
	\end{align*}
	because $\lim_{K \rightarrow \infty} f^{(K)}(\alpha, n) \leq \xi(\alpha, n)$ according to~(\ref{eq:lemApproxExpXiDef}).
\end{proof}

We now turn to a lemma focusing on the accuracy of a polynomial approximation of $ \alpha^{-1}(1- \exp(-\alpha))$.

\begin{amslem} \label{lem:approxExpSeries}
	For $0 < \alpha \leq 1$, $K \geq 1$ and $g: \lbrack 0, 1 \rbrack \rightarrow \mathbb{R} : \alpha \mapsto g(\alpha) = \alpha^{-1}(1- \exp(-\alpha))$,
	\begin{equation*}
		g(\alpha) = \sum_{k=0}^{K-1} \dfrac{\alpha^k}{(k+1)!} (-1)^{k} + R_{K-1}(\alpha)
	\end{equation*}
	where
	\begin{equation*}
		|R_{K-1}(\alpha)| \leq \dfrac{\alpha^{K}}{(K+1)!}
	\end{equation*}
\end{amslem}
\begin{proof}
	With $\ell(\alpha) := - \exp(-\alpha)$, it is easy to compute that 
	\begin{equation*}
		\dfrac{\mathrm{d}^k\ell}{\mathrm{d}\alpha^k} (x) = (-1)^{k+1} \exp(-x).
	\end{equation*}
	Thus,
	\begin{equation}\label{eq:maxDerivEll}
		\max_{x \in \lbrack 0, 1 \rbrack} \left| \dfrac{\mathrm{d}^{K+1}\ell}{\mathrm{d}\alpha^{K+1}} (x) \right| = 1.
	\end{equation} 
	Taylor's theorem \cite[Theorem~5.15]{rudin1964principles} shows that the $K$th-order Taylor polynomial of $\ell(\alpha)$ around zero has a remainder $R'_{K}(\alpha)$, for which $|R'_{K}(\alpha)| \leq \alpha^{K+1}/(K+1)!$ over $\alpha \in \lbrack 0, 1 \rbrack$ because of~(\ref{eq:maxDerivEll}). The desired $(K-1)$th-order polynomial approximation is:
	\begin{align*}
		\alpha^{-1}(1- \exp(-\alpha)) &= \alpha^{-1} \left( 1 - \sum_{k=0}^K \dfrac{\alpha^k}{k!}(-1)^k - R'_{K}(\alpha) \right) \\
		& = \sum_{k=0}^{K-1} \dfrac{\alpha^k}{(k+1)!} (-1)^{k} + R_{K-1}(\alpha),
	\end{align*}
	and the $(K-1)$th-order remainder is $R_{K-1}(\alpha) := -\alpha^{-1} R'_{K}(\alpha)$ and satisfies $|R_{K-1}(\alpha)| \leq \alpha^{K}/(K+1)!$.
\end{proof}

\subsection[\appendixname~\thesubsection]{Proof of Theorem~\ref{thm:approxThm}}
Using Theorem~\ref{thm:exactCollProp}, $\alpha = n/m$, $1/m = \alpha/n$ and Lemma~\ref{lem:approxExpLimit}, we derive
\begin{align}
	\dfrac{\mathbb{E} \left\lbrack Y^{(n,m)} \right\rbrack}{n} & = 1 - \dfrac{m}{n} \left( 1 - \left( \dfrac{m-1}{m} \right)^n \right) \nonumber \\
	& = 1 - \alpha^{-1} \left( 1 - \left( 1 - \alpha/n \right)^n \right) \nonumber \\
	& = 1 - \alpha^{-1} \left( 1 - \exp(-\alpha) F(\alpha, n) \right). \label{eq:proofTh2one}
\end{align}
For $n \geq 2$ and $\alpha < 1$, $\mu(\alpha, n) := \alpha^2 \sqrt{\dfrac{1}{n^2 - \alpha^2}  \left( \dfrac{\pi^2}{6} - 1 \right) }$ is monotonically decreasing with $n$ and monotonically increasing with $\alpha$, and it is approximately equal to $0.4637 < 1$ for $n=2$ and $\alpha = 1$. We shall use the inequality $1 - x \leq \exp (-x)$ (valid for $x<1$), with $x := \mu(\alpha, n)$---thereby implying $1 - \mu(\alpha, n) \leq \exp (-\mu(\alpha, n))$ because $\mu(\alpha, n) < 1$ for $n \geq 2$. Thus, from (\ref{eq:FanIneqApproxExpLimit}) of Lemma~\ref{lem:approxExpLimit}, we derive
\begin{equation} \label{eq:proofTh2two}
	1 - \alpha^2 \sqrt{\dfrac{1}{n^2 - \alpha^2}  \left(\dfrac{\pi^2}{6} - 1\right) } \leq F(\alpha, n) \leq 1.
\end{equation}
Therefore, by combining (\ref{eq:proofTh2one}) and (\ref{eq:proofTh2two}), we obtain
\begin{equation} \label{eq:upperIneq}
	\dfrac{\mathbb{E} \left\lbrack Y^{(n,m)} \right\rbrack}{n} \leq 1 - \alpha^{-1} \left. \left( 1 - \exp(-\alpha) F(\alpha, n) \right)\right|_{F(\alpha, n) = 1} = 1 - \alpha^{-1} \left( 1 - \exp(-\alpha) \right)
\end{equation}
and
\begin{align}
	\dfrac{\mathbb{E} \left\lbrack Y^{(n,m)} \right\rbrack}{n} & \geq 1 - \alpha^{-1} \left( 1 - \exp(-\alpha) F(\alpha, n) \right) \;\; \text{with} \;\; F(\alpha, n) = 1 - \alpha^2 \sqrt{\dfrac{1}{n^2 - \alpha^2}  \left(\dfrac{\pi^2}{6} - 1\right) } \nonumber \\
	& = 1 - \alpha^{-1} \left( 1 - \exp(-\alpha) \right) - \alpha^{-1} \alpha^2 \exp(- \alpha) \dfrac{1}{\sqrt{n^2 - \alpha^2}} \sqrt{\dfrac{\pi^2}{6}-1} \nonumber \\
	& = 1 - \alpha^{-1} \left( 1 - \exp(-\alpha) \right) - \exp(- \alpha) \sqrt{\dfrac{\alpha^2}{n^2 - \alpha^2}} \sqrt{\dfrac{\pi^2}{6}-1} \nonumber \\
	& \geq 1 - \alpha^{-1} \left( 1 - \exp(-\alpha) \right) - \sqrt{\dfrac{\alpha^2}{n^2 - \alpha^2}} \sqrt{\dfrac{\pi^2}{6}-1}, \label{eq:lowerIneq}
\end{align}
where the last line stems from  $-\exp(-\alpha) \geq - \exp(0) = -1$ for $\alpha \in \lbrack 0, 1 \rbrack$. As a result, combining (\ref{eq:upperIneq}) and (\ref{eq:lowerIneq}), we get
\begin{equation}
	- \sqrt{\dfrac{\alpha^2}{n^2 - \alpha^2}} \sqrt{\dfrac{\pi^2}{6}-1} \leq \dfrac{\mathbb{E} \left\lbrack Y^{(n,m)} \right\rbrack}{n} - \left(1 - \alpha^{-1} \left( 1 - \exp(-\alpha) \right)\right) \leq 0,
\end{equation}
which provides the bounds of the theorem (Equation~(\ref{eq:1stapprox})) for the error term $\delta(\alpha,n)$. Then, Lemma~\ref{lem:approxExpSeries} implies
\begin{align*}
	1 - \alpha^{-1} \left( 1 - \exp(-\alpha) \right) & = 1 - \sum_{k=0}^{K-1} \dfrac{\alpha^k}{(k+1)!} (-1)^{k} - R_{K-1}(\alpha) \\
	& = \sum_{k = 1}^{K-1} \dfrac{\alpha^k}{(k+1)!} (-1)^{k+1} - R_{K-1}(\alpha).
\end{align*}
Injecting this last result into~(\ref{eq:1stapprox}) proves (\ref{eq:2ndapprox}). Deriving~(\ref{eq:3rdapprox}) and~(\ref{eq:ROneTermRelativeIneq}) is straightforward.

\bibliographystyle{IEEEtran}
\bibliography{mybib}

\begin{thebibliography}{10}
\providecommand{\url}[1]{#1}
\csname url@samestyle\endcsname
\providecommand{\newblock}{\relax}
\providecommand{\bibinfo}[2]{#2}
\providecommand{\BIBentrySTDinterwordspacing}{\spaceskip=0pt\relax}
\providecommand{\BIBentryALTinterwordstretchfactor}{4}
\providecommand{\BIBentryALTinterwordspacing}{\spaceskip=\fontdimen2\font plus
\BIBentryALTinterwordstretchfactor\fontdimen3\font minus
  \fontdimen4\font\relax}
\providecommand{\BIBforeignlanguage}[2]{{%
\expandafter\ifx\csname l@#1\endcsname\relax
\typeout{** WARNING: IEEEtran.bst: No hyphenation pattern has been}%
\typeout{** loaded for the language `#1'. Using the pattern for}%
\typeout{** the default language instead.}%
\else
\language=\csname l@#1\endcsname
\fi
#2}}
\providecommand{\BIBdecl}{\relax}
\BIBdecl

\bibitem{martella2017current}
C.~Martella, J.~Li, C.~Conrado, and A.~Vermeeren, ``On current crowd management
  practices and the need for increased situation awareness, prediction, and
  intervention,'' \emph{Safety science}, vol.~91, pp. 381--393, 2017.

\bibitem{uras2020pma}
M.~Uras, R.~Cossu, E.~Ferrara, A.~Liotta, and L.~Atzori, ``{PmA: A real-world
  system for people mobility monitoring and analysis based on Wi-Fi probes},''
  \emph{Journal of Cleaner Production}, p. 122084, 2020.

\bibitem{determe2020forecasting}
J.-F. Determe, U.~Singh, F.~Horlin, and P.~De~Doncker, ``{Forecasting Crowd
  Counts With Wi-Fi Systems: Univariate, Non-Seasonal Models},'' \emph{IEEE
  Transactions on Intelligent Transportation Systems}, 2020.

\bibitem{singh2020crowd}
U.~Singh, J.-F. Determe, F.~Horlin, and P.~De~Doncker, ``{Crowd Forecasting
  based on WiFi Sensors and LSTM Neural Networks},'' \emph{IEEE Transactions on
  Instrumentation and Measurement}, 2020.

\bibitem{determe2022monitoring}
J.-F. Determe, S.~Azzagnuni, U.~Singh, F.~Horlin, and P.~De~Doncker,
  ``Monitoring large crowds with wifi: A privacy-preserving approach,''
  \emph{IEEE Systems Journal}, pp. 1--12, 2022.

\bibitem{dodis2013security}
Y.~Dodis, D.~Pointcheval, S.~Ruhault, D.~Vergniaud, and D.~Wichs, ``Security
  analysis of pseudo-random number generators with input: /dev/random is not
  robust,'' in \emph{Proceedings of the 2013 ACM SIGSAC conference on Computer
  \& communications security}, 2013, pp. 647--658.

\bibitem{stipvcevic2007quantum}
M.~Stip{\v{c}}evi{\'c} and B.~M. Rogina, ``Quantum random number generator
  based on photonic emission in semiconductors,'' \emph{Review of scientific
  instruments}, vol.~78, no.~4, p. 045104, 2007.

\bibitem{zheng20196}
Z.~Zheng, Y.~Zhang, W.~Huang, S.~Yu, and H.~Guo, ``6 gbps real-time optical
  quantum random number generator based on vacuum fluctuation,'' \emph{Review
  of Scientific Instruments}, vol.~90, no.~4, p. 043105, 2019.

\bibitem{demir2014analysing}
L.~Demir, M.~Cunche, and C.~Lauradoux, ``Analysing the privacy policies of
  {Wi-Fi} trackers,'' in \emph{Proceedings of the 2014 workshop on physical
  analytics}, 2014, pp. 39--44.

\bibitem{leach2005universally}
P.~Leach, M.~Mealling, and R.~Salz, ``{A universally unique identifier (UUID)
  URN namespace},'' 2005.

\bibitem{demir2017pitfalls}
L.~Demir, A.~Kumar, M.~Cunche, and C.~Lauradoux, ``The pitfalls of hashing for
  privacy,'' \emph{IEEE Communications Surveys \& Tutorials}, vol.~20, no.~1,
  pp. 551--565, 2017.

\bibitem{marx2018hashing}
M.~Marx, E.~Zimmer, T.~Mueller, M.~Blochberger, and H.~Federrath, ``Hashing of
  personally identifiable information is not sufficient,'' \emph{SICHERHEIT
  2018}, 2018.

\bibitem{fuxjaeger2016towards}
P.~Fuxjaeger, S.~Ruehrup, T.~Paulin, and B.~Rainer, ``Towards
  privacy-preserving {Wi-Fi} monitoring for road traffic analysis,'' \emph{IEEE
  Intelligent Transportation Systems Magazine}, vol.~8, no.~3, pp. 63--74,
  2016.

\bibitem{ali2020practical}
\BIBentryALTinterwordspacing
J.~Ali and V.~Dyo, ``Practical hash-based anonymity for {MAC} addresses,'' in
  \emph{Proceedings of the 17th International Joint Conference on e-Business
  and Telecommunications, {ICETE} 2020 - Volume 2: SECRYPT, Lieusaint, Paris,
  France, July 8-10, 2020}, P.~Samarati, S.~D.~C. di~Vimercati, M.~S. Obaidat,
  and J.~Ben{-}Othman, Eds.\hskip 1em plus 0.5em minus 0.4em\relax ScitePress,
  2020, pp. 572--579. [Online]. Available:
  \url{https://doi.org/10.5220/0009825105720579}
\BIBentrySTDinterwordspacing

\bibitem{hong2018crowdprobe}
H.~Hong, G.~D. De~Silva, and M.~C. Chan, ``Crowdprobe: non-invasive crowd
  monitoring with {W}i-{F}i probe,'' \emph{Proceedings of the ACM on
  Interactive, Mobile, Wearable and Ubiquitous Technologies}, vol.~2, no.~3,
  pp. 1--23, 2018.

\bibitem{potorti2018localising}
F.~Potort{\`\i}, A.~Crivello, M.~Girolami, P.~Barsocchi, and E.~Traficante,
  ``Localising crowds through wi-fi probes,'' \emph{Ad Hoc Networks}, vol.~75,
  pp. 87--97, 2018.

\bibitem{rtx2080fesha256benchmark}
``{Nvidia RTX 2080 SUPER FE Hashcat Benchmarks},''
  \url{https://gist.github.com/epixoip/47098d25f171ec1808b519615be1b90d},
  accessed: 2020-08-13.

\bibitem{provos1999future}
N.~Provos and D.~Mazieres, ``A future-adaptable password scheme.'' in
  \emph{USENIX Annual Technical Conference, FREENIX Track}, 1999, pp. 81--91.

\bibitem{biryukov2016argon2}
A.~Biryukov, D.~Dinu, and D.~Khovratovich, ``Argon2: new generation of
  memory-hard functions for password hashing and other applications,'' in
  \emph{2016 IEEE European Symposium on Security and Privacy (EuroS\&P)}.\hskip
  1em plus 0.5em minus 0.4em\relax IEEE, 2016, pp. 292--302.

\bibitem{mivskinis2014timing}
R.~Mi{\v{s}}kinis, D.~Jokubauskis, D.~Smirnov, E.~Urba, B.~Maly{\v{s}}ko,
  B.~Dzindzel{\.e}ta, and K.~Svirskas, ``{Timing over a 4G (LTE) mobile
  network},'' in \emph{2014 European Frequency and Time Forum (EFTF)}.\hskip
  1em plus 0.5em minus 0.4em\relax IEEE, 2014, pp. 491--493.

\bibitem{menezes1996handbook}
A.~J. Menezes, J.~Katz, P.~C. Van~Oorschot, and S.~A. Vanstone, \emph{Handbook
  of applied cryptography}.\hskip 1em plus 0.5em minus 0.4em\relax CRC press,
  1996.

\bibitem{rudin1964principles}
W.~Rudin \emph{et~al.}, \emph{Principles of mathematical analysis}.\hskip 1em
  plus 0.5em minus 0.4em\relax McGraw-hill New York, 1964, vol.~3.

\end{thebibliography}

\end{document}